\documentclass[a4paper,UKenglish,cleveref, cleveref, thm-restate]{lipics-v2019}

\usepackage{amsthm}
\usepackage{xspace}
\usepackage{macros}
\usepackage{coqtheorem}
\usepackage{tikz-cd}
\usepackage{adjustbox}
\usepackage{hyperref}
\usepackage{cite}

\setBaseUrl{{https://uds-psl.github.io/churchs-thesis-coq/coq/website/SyntheticComputability.}}

\title{Church's thesis and related axioms\\ in Coq's type theory} 

\titlerunning{Church's thesis and related axioms in Coq's type theory}

\author{Yannick Forster}{Saarland University, Saarland Informatics Campus, Saarbrücken, Germany}{forster@cs.uni-saarland.de}{https://orcid.org/0000-0002-8676-9819}{}

\authorrunning{Y. Forster} %

\Copyright{Yannick Forster} %



\category{} %

\relatedversion{} %

\supplement{\url{https://github.com/uds-psl/churchs-thesis-coq}}%

\acknowledgements{I want to thank Gert Smolka, Andrej Dudenhefner, Dominik Kirst, and Dominique Larchey-Wendling for discussions and feedback on drafts of this paper. Special thanks go to the anonymous reviewers for their helpful ideas, constructive comments, and editorial suggestions.}%

\nolinenumbers %

\hideLIPIcs  %

\EventEditors{Christel Baier and Jean Goubault-Larrecq}
\EventNoEds{2}
\EventLongTitle{29th EACSL Annual Conference on Computer Science Logic (CSL 2021)}
\EventShortTitle{CSL 2021}
\EventAcronym{CSL}
\EventYear{2021}
\EventDate{January 25--28, 2021}
\EventLocation{Ljubljana, Slovenia (Virtual Conference)}
\EventLogo{}
\SeriesVolume{183}
\ArticleNo{39}

\newcommand\definedas{:=}
\newcommand{\lipicsNumber}[1]{\text{{\color{lipicsGray}\sffamily\textbf{#1}}}}

\newcommand{\introterm}[1]{\emph{#1}}

\newcommand\CC{{\textnormal{\textsf{-ACC}}}}

\definecolor{sb@dcyan}{rgb}{0.31,0.31,0.33}%

\sfcommand{L}

\begin{document}

\maketitle

\begin{abstract}
  ``Church's thesis'' ($\mathsf{CT}$) as an axiom in constructive logic states that every total function of type $\mathbb{N} \to \mathbb{N}$ is computable, i.e.\ definable in a model of computation.
  $\mathsf{CT}$ is inconsistent both in classical mathematics and in Brouwer's intuitionism since it contradicts weak Kőnig's lemma and the fan theorem, respectively.
  Recently, $\mathsf{CT}$ was proved consistent for (univalent) constructive type theory.

  Since neither weak Kőnig's lemma nor the fan theorem is a consequence of just logical axioms or just choice-like axioms assumed in constructive logic, it seems likely that $\CT$ is inconsistent only with a combination of classical logic and choice axioms.
  We study consequences of $\mathsf{CT}$ and its relation to several classes of axioms in Coq's type theory, a constructive type theory with a universe of propositions which proves neither classical logical axioms nor strong choice axioms.

  We thereby provide a partial answer to the question as to which axioms may preserve computational intuitions inherent to type theory, and which~certainly~do~not.
  The paper can also be read as a broad survey of axioms in type theory, with all results mechanised in the Coq proof assistant.
\end{abstract}

\section{Introduction}%

The intuition that the concept of a constructively defined function and a computable function can be identified is prevalent in intuitionistic logic since the advent of recursion theory and is maybe most natural in constructive type theory, where computation is primitive. %

A formalisation of the intuition is the axiom $\CT$ (``Church's thesis''), stating that every function is computable, i.e.\ definable in a model of computation.
$\CT$ is well-studied as part of Russian~constructivism~\cite{markov1954theory} and in the field of constructive reverse mathematics~\cite{IshiharaCRM,dienerConstructiveReverseMathematics2020}.%

$\CT$ allows proving results of recursion theory without extensive references to a model of computation, since one can reason with functions instead.
While such synthethic developments of computability theory~\cite{bridges1987varieties,richman1983church,BauerSyntCT} can be carried out in principle without assuming any axioms~\cite{forster2019synthetic}, assuming $\CT$ allows stronger results:
$\CT$ essentially provides a universal machine w.r.t.\ all functions in the logic, allowing to show the non-existence of certain deciding functions -- whose existence is logically independent with no axioms present.

It is easy to see that $\CT$ is in conflict with traditional classical mathematics, since the law of excluded middle $\LEM$ together with a form of the axiom of countable choice $\AC_{\nat,\nat}$ allows the definition of non-computable functions~\cite{troelstra1988constructivism}.
This observation can be sharpened in various ways:
To define a non-computable function directly, the weak limited principle of omniscience $\WLPO$ and the countable unique choice axiom $\AUC_{\nat,\bool}$ suffice.
Alternatively, Kleene noticed that there is a decidable tree predicate with infinitely many nodes but no computable infinite path~\cite{kleene1953recursive}.
If functions and computable functions are identified via $\CT$, a Kleene tree is in conflict with weak Kőnig's lemma $\WKL$ and with Brouwer's fan theorem.

It is however well-known that $\CT$ is consistent in Heyting arithmetic with Markov's principle $\MP$~\cite{kleene1945interpretation} which given $\CT$ states that termination of computation is stable under double negation.
Recently, Swan and Uemura~\cite{swan2019church} proved that $\CT$ is consistent in univalent type theory with propositional truncation and $\MP$.

While predicative Martin-Löf type theory as formalisation of Bishop's constructive mathematics proves the full axiom of choice $\AC$, univalent type theory usually only proves the axiom of unique choice $\AUC$.
But since $\AUC_{\nat,\bool}$ suffices to show that $\LEM$ implies $\neg \CT$, classical logic is incompatible with $\CT$ in both predicative and in univalent type theory.

In the (polymorphic) calculus of (cumulative) inductive constructions, a constructive type theory with a separate, impredicative universe of propositions as implemented by the proof assistant Coq~\cite{Coq}, none of $\AC$, $\AUC$, and $\AUC_{\nat,\bool}$ are provable.
This is because large eliminations on existential quantifications are not allowed in general~\cite{paulin1993inductive}, meaning one can not recover a function in general from a proof of $\forall x.\exists y.\;R x y$.
However, choice axioms as well al $\LEM$ can be consistently assumed in Coq's type theory~\cite{werner1997sets}. %
Furthermore, it seems likely that the consistency proof for $\CT$ in~\cite{swan2019church} can be adapted for Coq's type theory.

This puts Coq's type theory in a special position:
Since to disprove $\CT$ one needs a (weak) classical logical axiom and a (weak) choice axiom, assuming just classical logical axioms or just choice axioms might be consistent with $\CT$.
This paper is intended to serve as a preliminary report towards this consistency question, approximating it by surveying results from intuitionistic logic and constructive reverse mathematics in constructive type theory with a separate universe of propositions, with a special focus on $\CT$ and other axioms based on notions from computability theory.
Specifically, we discuss these propositional axioms:
\sfcommand{Homeo}
\begin{itemize}
\item computational enumerability axioms ($\EA, \EPF$) and Kleene trees (\KT) in \Cref{sec:kleene}
\item extensionality axioms like functional extensionality ($\Fext$), propositional extensionality (\Pext), and proof irrelevance (\PI) in \Cref{sec:ext}
\item classical logical axioms like the principle of excluded middle (\PEM, \WLEM), independence of premises (\IP), and limited principles of omniscience (\LPO, \WLPO, \LLPO) in \Cref{sec:class}
\item axioms of Russian constructivism like Markov's principle (\MP) in \Cref{sec:russ}%
\item choice axioms like the axiom of choice (\AC), countable choice ($\ACC$, $\AC_{\nat,\nat}$, $\AC_{\nat,\bool}$), dependent choice ($\ADC$), and unique choice ($\AUC, \AUC_{\nat,\bool}$) in~\Cref{sec:choice}
\item axioms on trees like weak Kőnig's lemma (\WKL) and the fan theorem (\FAN) in \Cref{sec:trees}
\item axioms regarding continuity and Brouwerian principles ($\Homeo$, %
  \Cont, \WCN) in \Cref{sec:cont}
\end{itemize}
The following hyper-linked diagram displays provable implications and incompatible axioms.
\begin{figure}[h]
  \centering
\adjustbox{scale=0.85}{
\begin{tikzcd}[column sep=tiny]
  \DNE \ar[d]\ar[r,leftrightarrow] & \LEM \ar[r]\ar[d] & \DGP \ar[r]\ar[dr] & \WLEM\ar[dl] & \ADC \ar[d] & \AC \ar[l]\ar[d]\ar[llll, swap, bend right=10, "\Fext"] \\
  \MP & \LPO\ar[l]\ar[r] & \WLPO\ar[r] \ar[l,"\MP", bend right]& \LLPO \ar[l,bend left, "\PFP"] \arrow[d, "\compl\semidecidable\texttt{-}\AC_{\nat,\bool}"{name=A}] & \ACC \ar[d] & \AC_{\nat \to \nat,\nat} \ar[dl]\ar[dddl,bend left=65,red,dashed,dash,"\color{black}\Fext"] \ar[d,swap,"\Cont"{name=cont}] \\
  & {\Homeo(\bool^\nat,\nat^\nat)} \ar[dr] & {\Homeo(\nat^\nat,\bool^\nat)} & \WKL \ar[d] & \AC_{\nat,\nat} \ar[d] \arrow[to=A] & \WCN \ar[ddl,red,dash,dashed]\\ %
  & & \KT \ar[lu] \ar[u] & \FAN \ar[l,dashed,red,dash] & \AUC_{\nat,\bool} \ar[ld,swap,dashed,red,dash,"\color{black}\WLPO"]  \\
  & & \EPF \ar[u] & \EA \ar[l,leftrightarrow] & \CT \ar[l]
\end{tikzcd}
}
\vspace{-1\baselineskip}
\caption{Overview of results. $\rightarrow$ are implications, 	 \protect\tikz[baseline]{\protect\draw[line width=0.2mm,densely dashed, red] (0,.8ex)--++(0.5,0);}~denotes incompatible axioms.}
\end{figure}
\newpage
All results in this paper are mechanised in the Coq proof assistant and the proof scripts are accessible at
\url{https://github.com/uds-psl/churchs-thesis-coq}.
The statements in this document are hyperlinked to their Coq proof, indicated by a \includegraphics[height=1em]{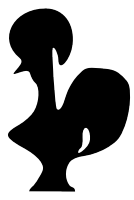}-symbol.

\textbf{Outline.}~~
\Cref{sec:prelim} establishes necessary preliminaries regarding Coq's type theory and introduces the notions of (synthetic) decidability, enumerability, and semi-decidability.
\Cref{sec:CT} introduces $\CT$ formally, together with the related synthetic axioms $\EA$ and $\EPF$.
\Cref{sec:synth} contains undecidability proofs %
based on $\CT$.
\Cref{sec:kleene} introduces decidable binary trees and constructs a Kleene tree. %
The connection of $\CT$ to the classes of axioms as listed above is surveyed in~\Cref{sec:ext,sec:class,sec:russ,sec:choice,sec:trees,sec:cont}.
\Cref{sec:conclusion} contains concluding remarks.

\section{Preliminaries}
\label{sec:prelim}

We work in the polymorphic calculus of cumulative inductive constructions as implemented by the Coq proof assistant~\cite{Coq}, which we will refer to as ``Coq's type theory''.
The calculus is a constructive type theory with a cumulative hierarchy of types $\Type_i$ (where $i$ is a natural number, but we leave out the index from now on), an impredicative universe of propositions $\Prop \subseteq \Type$, and inductive types in every universe.
\mbox{The inductive types of interest in this paper are}\label{def:map}\label{def:Some}\label{def:None}
\begin{align*}
  n : \nat &::= 0 \mid \succN\,n &
  b : \bool &::= \bfalse \mid \btrue \\
  o : \option A &::= \Noneintro \mid \Someintro a \quad \textit{where $a : A$} &
  l : \List A &::= [] \mid a :: l \quad \textit{where $a : A$} \\
  A + B &:= \inlintro a \mid \inrintro b \quad \textit{where $a : A$ and $b : B$}  &
  A \times B &:= (a, b) \quad \textit{where $a : A$ and $b : B$} 
\end{align*}
\newcommand\sigpair[2]{(#1,#2)}%
One can easily construct a pairing function $\langle \_ \,,\, \_ \rangle : \nat \to \nat \to \nat$ and for all $f : \nat \to \nat \to X$ an inverse construction $\lambda \langle  n, m \rangle .\; f n m$ of type $\nat \to X$ s.t.\ $(\lambda \langle n, m \rangle.\; f n m) \langle n, m \rangle = f n m$.

We write $n =_\bool m$ for the boolean equality decider on $\nat$, and $\neg_\bool$ for boolean negation.

If $l : \List A$ then $l[n] : \option A$ denotes the $n$-th element of $l$.
If $n < \length l$ we can assume $l[n] : A$.

We write $\forall x : X.\; A x$ for both dependent functions and logical universal quantification, $\exists x : X.\; A x$ where $A : X \to \Prop$ for existential quantification and $\Sigma x : X.\; A x$ where $A : X \to \Type$ for dependent pairs, with elements $\sigpair x y$.
Dependent pairs can be eliminated into arbitrary types, i.e.\ there is an elimination principle of type $\forall p : (\Sigma x.\; A x) \to \Type.\; (\forall (x : X) (y : A x). \; p \sigpair x y) \to \forall (s : \Sigma x.\; A x).\; p s.$
We call such a principle eliminating a proposition into arbitrary types a \introterm{large elimination principle}, following the terminology ``large elimination'' for Coq's case analysis construct \texttt{match}~\cite{paulin1993inductive}.
Crucially, Coq's type theory proves a large elimination principle for the falsity proposition $\bot$, i.e. explosion applies to arbitrary types: $\forall A : \Type.\; \bot \to A$.
In contrast, existential quantification can only be eliminated for $p : (\exists x.\; A x) \to \Prop$, but the following more specific large elimination principle is provable:%
\setCoqFilename{Axioms.axioms}%
\begin{lemma}[][mu_nat]
  There is a guarded minimisation function %
  $\mu_\nat$ of the following type: \[\mu_{\nat} : \forall f : \nat \to \bool.\;(\exists n.\; f n = \btrue) \to \Sigma n .\; f n = \btrue \land \forall m.\; f m = \btrue \to m \geq n.\]
\end{lemma}
There are various implementations of such a minimisation function in \href{https://coq.inria.fr/library/Coq.Logic.ConstructiveEpsilon.html}{Coq's Standard Library}\rlap.\footnote{The idea was conceived independently by Benjamin Werner and Jean-François Monin in the 1990s.}
One uses a (recursive) large elimination principle for the accessibility predicate, see e.g.\ \cite[\S 2.7, \S 4.1, \S 4.2]{larcheybraga} and~\cite[\S 14.2.3, \S 15.4]{bertot2013interactive} for a contemporary overview how to implement large eliminations principles.
We will not need any other large elimination principle in this paper.
A restriction of large elimination in general is necessary for consistency of Coq~\cite{coquand:inria-00075471}.
As a by-product, the computational universe $\Type$ is separated from the logical universe $\Prop$, allowing classical logic in $\Prop$ to be assumed while the computational intuitions for $\Type$ remain intact.

\subsection{Partial Functions}

All definable functions in type theory are total by definition.
To model partiality, one often resorts to functional relations $R : A \to B \to \Prop$ or step-indexed functions $A \to \nat \to \option B$, as for instance pioneered by Richman~\cite{richman1983church} in constructive logic, see e.g.~\cite{escard_et_al:LIPIcs:2017:7682} for a comprehensive overview.%

For our purpose, we simply assume a type $\partial A$ for $A : \Type$ and a definedness relation $\hasvalue : \partial A \to A \to \Prop$ and write $A \pfun B$ for $A \to \partial B$.
We assume monadic structure for $\partial$ ($\retintro$ and $\bind$),\label{def:ret}
an undefined value ($\textsf{undef}$),
a minimisation operation ($\mu$), and a step-indexed evaluator ($\textsf{seval}$).
The operations and their specifications are listed in \Cref{fig:partial}.
 
\begin{figure}\centering
  
  \begin{tabular}{l l l}
    $\partial A : \mathbb{T}$ & partial values over $A : \mathbb{T}$ & \\
    $\hasvalue : \partial A \to A \to \Prop$ & definedness of values & $x \hasvalue a_1 \to x \hasvalue a_2 \to a_1 = a_2$ \\
    $\ter {(x : \partial A)} : \Prop$ & & $\ter x := \exists a.\; x \hasvalue a$ \\ 
    $\equiv_{\partial A} \,: \partial A \to \partial A \to \Prop $ & equivalence & $x \equiv_{\partial A} y := (\forall a.\; x \hasvalue a \leftrightarrow y \hasvalue a)$ \\
    \hline
    $\ret : A \to \partial A$ & monadic return & $\ret a \hasvalue a$ \\
    $\textsf{undef} : \partial A$ & undefined value & $\nexists a.\textsf{undef} \hasvalue a$ \\
    $\bind : \partial A \to (A \to \partial B) \to \partial B$ & monadic bind & $x \bind f \hasvalue b \leftrightarrow (\exists a.\; x \hasvalue a \land f a \hasvalue b)$ \\
    \hline
    $\mu : (\nat \to \bool) \to \partial \nat$ & unbounded search & $\begin{array}{ll}\mu f \hasvalue n \leftrightarrow f n = \btrue\, \land \\ \hspace{1.5cm}\forall m < n.\;f m = \bfalse\end{array}$ \\
    \hline
    $\textsf{seval} : \partial A \to \nat \to \option A$ & step-indexed evaluation & $x \hasvalue a \leftrightarrow \exists n.\;\textsf{seval}\, x n = \Some a$ \\
  \end{tabular}

  \caption{A monad for partial values}
  \label{fig:partial}
\end{figure}

\subsection{Equivalence relations on functions}

Besides intensional equality ($=$), we will consider other more extensional equivalence relations in this paper.
For instance, extensional equality of functions $f, g$ ($\forall x.\;f x = g x$), extensional equivalence of predicates $p, q$ ($\forall x.\;p x \leftrightarrow q x$), or range equivalence of functions $f, g$ ($\forall x.\; (\exists y.\;f y = x) \leftrightarrow (\exists y.\;g y = x)$).
We will denote all of these equivalence relations with the symbol $\equiv$ and indicate what is meant by an index.
For discrete $X$ (e.g.\ $\nat$, $\option \nat$, $\List \bool$, \dots), $\equivwrt{X}$ denotes equality, $\equivwrt{\Prop}$ denotes logical equivalence, $\equivwrt{A \to B}$ denotes an extensional lift of $\equivwrt{B}$, $\equivwrt{A \to \Prop}$ denotes extensional equivalence, and $\equivwrt{\textsf{ran}}$ denotes range equivalence.

Assuming the existence of surjections $A \to (A \to B)$ may or may not be consistent, depending on the particular equivalence relation.
We introduce the notion of \emph{surjection w.r.t.\ $\equivwrt{B}$} as $\forall b : B.~\exists a : A. f a \equivwrt{B} b$.
We call a function $f : A \to B$ an \emph{injection w.r.t.\ $\equivwrt{A}$ and $\equivwrt{B}$} if $\forall a_1 a_2.~f a_1 \equivwrt{B} f a_2 \to a_1 \equivwrt{A} a_2$ and a \emph{bijection} if it is an injection and surjection.

One formulation of Cantor's theorem is that there is no surjection $\nat \to (\nat \to \nat)$ w.r.t.\ $=$.
However, the same proof can be used for the following strengthening of Cantor's theorem:

\setCoqFilename{Axioms.axioms}
\begin{fact}[Cantor][Cantor]
  There is no surjection $\nat \to (\nat \to \nat)$ w.r.t.~$\equivwrt{\nat \to \nat}$.
\end{fact}
\subsection{Decidability, Semi-decidability, Enumerability, Reducibility}

We define decidability, (co-)semi-decidability, and enumerability for predicates $p : X \to \Prop$:
\[\begin{array}{lll@{\hspace{1cm}}r}
  \decidable p & := \exists f : X \to \bool.&\forall x.~p x \leftrightarrow f x = \textsf{true} &\text{(``$p$ is decidable'')} \\
  \semidecidable p & := \exists f : X \to \nat \to \bool.&\forall x.~p x \leftrightarrow \exists n. f x n = \btrue &\text{(``$p$ is semi-decidable'')} \\
  \compl\semidecidable p & := \exists f : X \to \nat \to \bool.&\forall x.~p x \leftrightarrow \forall n. f x n = \bfalse &\text{(``$p$ is co-semi-decidable'')} \\
  \enumerable p & := \exists f : \nat \to \option X.&\forall x.~p x \leftrightarrow \exists n. f n = \Some x &\text{(``$p$ is enumerable'')} 
\end{array}\]

Although all notions are defined on unary predicates, we use them on $n$-ary relations via (implicit) uncurrying.
We write $\compl p$ for the complement $\lambda x.\;\neg p x$ of $p$.
We call a type $X$ \introterm{discrete} if its equality relation $=_{X}$ is decidable and \introterm{enumerable} if the predicate $\lambda x.\top$ is enumerable.

Traditionally, propositions $P$ s.t.\ $P \leftrightarrow (\exists n.\; f n =\btrue)$ for some $f$ are often called $\Sigma^0_1$ or ``simply existential'', and $P$ s.t.\ $P \leftrightarrow (\forall n.\; f n = \bfalse)$ are called $\Pi^0_1$ or ``simply universal''.
Semi-decidable predicates are pointwise $\Sigma^0_1$, and co-semi-decidable predicates are pointwise $\Pi^0_1$.
Note that neither $\compl\semidecidable p \to \semidecidable \compl p$ nor the converse is provable, only the following connections:

\setCoqFilename{Synthetic.SemiDecidabilityFacts}%
\begin{lemma}[][decidable_semi_decidable]
  The following hold:
  \begin{enumerate}
    \coqitem[decidable_semi_decidable] Decidable predicates are semi-decidable and co-semi-decidable.
    \setCoqFilename{Synthetic.EnumerabilityFacts}%
    \coqitem[semi_decidable_enumerable] Semi-decidable predicates on enumerable types are enumerable.
    \coqitem[enumerable_semi_decidable] Enumerable predicates on discrete types are semi-decidable.
    \setCoqFilename{Synthetic.SemiDecidabilityFacts}%
  \coqitem[sdec_co_sdec_comp] The complement of semi-decidable predicates is co-semi-decidable.
  \end{enumerate}
\end{lemma}

\setCoqFilename{Synthetic.DecidabilityFacts}
\begin{lemma}[][dec_compl]
  Decidable predicates are closed under complementation.
  Decidable, enumerable, and semi-decidable predicates are closed under (pointwise) conjunction and disjunction.
\end{lemma}

\section{Church's thesis in type theory}
\label{sec:CT}

Church's thesis for total functions ($\CT$) states that every function of type $\nat \to \nat$ is algorithmic.
Thus $\CT$ is a relativisation of the function space $\nat \to \nat$ w.r.t.\ a given (Turing-complete) %
model of computation%
, reminiscent of the axiom $V = L$ in set theory~\cite{kreisel1970church}.

We first define $\CT$ by abstracting away from a concrete model of computation and work with an \introterm{abstract model of computation}, consisting of an \emph{abstract computation function} $T c x n$ (with $T : \mathbb{N} \to \mathbb{N} \to \mathbb{N} \to \option \nat$),
assigning to a code $c$ (to be interpreted as the code of a partial recursive function in a model of computation), an input number $x$, and a step index $n$ an output number $y$ if the code terminates in $n$ steps on $x$ with value $y$. %
The function $T c x$ is assumed to be monotonic, i.e.\ increasing the step index does not change the potential value:
\[ T c x n_1 = \Some y \to \forall n_2 \geq n_1.~T c x n_2 = \Some y. \]

Based on $T$ we define a computability relation between $c : \nat$ and $f : \nat \to \nat$:
\[ c \sim f := \forall x. \exists n.~T c x n = \Some (f x). \] %

Since $T$ is monotonic, $\sim$ is extensional, i.e.\ $n \sim f_1 \to n \sim f_2 \to \forall x.\;f_1 x = f_2 x$.
We define Church's thesis for total functions relative to an abstract computation function~$T$:%
\label{def:CT}%
\[ \CTintro_T := \forall f : \nat \to \nat.\exists n : \nat.~n\sim f \]

Note that $\CT_T$ is clearly not consistent for every choice of $T$.
If we write $\CT$ without index, we mean $T$ to be the step-indexed evaluation function of a concrete, Turing-complete model of computation.
For the mechanisation  we could for instance pick the equivalent models of Turing machines~\cite{forster2020verified}, $\lambda$-calculus~\cite{forster2019call}, $\mu$-recursive functions~\cite{larchey2017typing}, or register machines~\cite{forster2019certified,larchey2020hilbert}.
It seems likely that the consistency proof of $\CT$ in \cite{swan2019church} can be adapted to Coq.

Since specific properties of the model of computation are not needed, we develop and mechanise all results of this paper parameterised in an arbitrary $T$.
Thus, we could also state all results in terms of a fully synthetic Church's thesis axiom $\Sigma T.\CT_T$.
\begin{fact}
  $\CT \to \Sigma T.\CT_T$
\end{fact}

Note that the implication is strict:
An abstract computation function does not rule out oracles for e.g.\ the halting problem of Turing machines, whereas $\CT$ -- with $T$ defined in terms of a standard, Turing-complete model of computation -- proves the undecidability of the Turing machine halting problem.

\subsection{Bauer's enumerability axiom $\EA$}

In proofs of theorems with $\CT_T$ as assumption, $T$ can be used as replacement for a \emph{universal machine}.
Bauer~\cite{BauerSyntCT} develops computability theory synthetically using the axiom ``the set of enumerable sets of natural numbers is enumerable'', which is equivalent to $\Sigma T.\CT_T$ and thus strictly weaker than $\CT$, but can also be used in place of a universal machine.
We introduce Bauer's axiom in our setting as $\EA'$ and immediately introduce a strengthening~$\EA$ s.t.\ $(\Sigma T.\CT_T) \leftrightarrow \EA$ and $\EA \to \EA'$:
\newcommand\W{\mathcal{W}}%
\[ \EAintro' := \Sigma \W : \nat \to (\nat \to \mathbb{P}).\forall p : \nat \to \mathbb{P}.~
  \enumerable p \leftrightarrow \exists c.~\W c \equivwrt{\nat \to \Prop} p\]
That is, $\EA'$ states that there is an enumerator $\W$ of all enumerable predicates, up to extensionality.
In contrast, $\EA$ poses the existence of an enumerator of all possible enumerators, up to range equivalence:
\label{def:EA}%
\[  \EAintro := \Sigma \varphi : \nat \to (\nat \to \option\nat).\forall f : \nat \to \option \nat. \exists c.\; \varphi c \equiv_{\textsf{ran}} f \]

That is, $\varphi$ is a surjection w.r.t.\ range equivalence $f \equiv_{\textsf{ran}} g$, where $\varphi c \equiv_{\textsf{ran}} f \leftrightarrow \forall x.(\exists n.\varphi c n = \Some x) \leftrightarrow (\exists n. f n = \Some x)$.

Note the two different roles of natural numbers in the two axioms:
If we would consider predicates over a general type $X$ we would have $\W : \nat \to (X \to \Prop)$ and $\varphi : \nat \to (\nat \to \option X)$, i.e.\ $\W c$ would be an enumerable predicate and $\varphi c$ an enumerator of a predicate $X \to \Prop$.

We start by proving $\CT_T \to \EA$ by constructing $\varphi$ from an arbitrary $T$:
\[\varphi c \langle n,m \rangle \definedas \iteis{T c n m}{\Some x}{\succN x}{0}\]
\vspace{-1.6\baselineskip}
\setCoqFilename{Axioms.axioms}
\begin{lemma}[][CT_to_EA']
  If $\CT_T$ then $\forall f : \nat \to \option \nat. \exists c.\; \varphi c \equiv_{\textsf{ran}} f$.
\end{lemma}
\begin{proof}
  The direction from left to right to establish $\equiv_{\textsf{ran}}$ is based on the fact that if $T c x n_1 = \Some y_1$ and $T c x n_2 = \Some y_2$ then $y_1 = y_2$.
  The other direction is straightforward.
\end{proof}

\begin{theorem}[][CT_to_EA]
  $\forall T.\;\CT_T \to \EA$
\end{theorem}

We now prove $\EA \to \EA'$ by constructing $\W$ from $\varphi$: $ \W c x := \exists n. \varphi c n = \Some x$.
\begin{lemma}[][EA_to_EA'_prf]
  If $\EA$ then $\forall p : \nat \to \mathbb{P}.~\enumerable p \leftrightarrow \exists c.~\W c \equivwrt{\nat \to \Prop} p$.
\end{lemma}
\begin{proof}$\begin{aligned}[t] \enumerable p
    \leftrightarrow~ & \exists f : \nat \to \option \nat.\forall x.\; p x \leftrightarrow \exists n.\; f n = \Some x &&\text{(def. $\enumerable$)} \\
    \leftrightarrow~ & \exists c.\forall x.\; p x \leftrightarrow \exists n.\; \varphi c n = \Some x &&\text{(\EA) } \\
    \leftrightarrow~ & \exists c.\;\W c \equivwrt{\nat\to\Prop} p &&\text{(def.
      $\equivwrt{\nat\to\Prop})$} \qquad\qquad\qquad\qquad~~\popQED
  \end{aligned}$
\end{proof}

\begin{theorem}[][EA_to_EA']
  $\EA \to \EA'$
\end{theorem}

\subsection{Richman's Enumerability of Partial Functions $\EPF$}

Richman~\cite{richman1983church} introduces a different purely synthetic axiom as replacement for a universal machine and assumes that ``partial functions are countable'', which is equivalent to $\EA$.
\label{def:EPF}
\[ \EPFintro := \Sigma e : \nat \to (\nat \pfun \nat).\forall f : \nat \pfun \nat. \exists n.\;e n \equivwrt{\nat\pfun\nat} f \]

\begin{theorem}[][EPF_to_EA]
  $\EPF \to \EA$
\end{theorem}
\begin{proof}
  Let $e$ be given.
  $\varphi c \langle n, m \rangle := \textsf{seval}\ (e c n)\ m$ is the wanted enumerator.
\end{proof}

\begin{theorem}[][EA_to_EPF]
  $\EA \to \EPF$
\end{theorem}
\begin{proof}
  Let $\varphi$ be given. Then
  \begin{align*}
    e c x := &\left( \mu \left(\lambda n.~\iteis{\varphi c n}{\Some \langle x', y' \rangle}{x =_{\bool} x'}{\bfalse}\right)\right) \bind \\ & \quad \lambda n.~\iteis{\varphi c n}{\Some \langle x', y' \rangle}{\ret y'}{\textsf{undef}}
  \end{align*}
  is the wanted enumerator.
\end{proof}

$\EPF$ implies the fully synthetic version of $\CT$:

\begin{lemma}[][EPF_to_CT]
  \label{lem:epf_induces_model}%
  $\EPF \to \Sigma T.\;\CT_T$
\end{lemma}
\begin{proof}
  Assume $e : \nat \to (\nat \pfun \nat)$ surjective w.r.t.\ $\equivwrt{\nat\pfun\nat}$.
  Define $T c x n := \mathsf{seval}~(e c x)~n$.
  It is straightforward to prove that $T$ is monotonic and that $\CT$ holds.
\end{proof}

The axiom $\EPF$ can be weakened to cover just boolean functions:
\[\EPF_\bool := \Sigma e : \nat \to (\nat \pfun \bool).\forall f : \nat \pfun \bool. \exists n.\;e n \equivwrt{\nat\pfun\bool} f\]

\begin{lemma}[][EPF_to_EPF_bool]
  $\EPF \to \EPF_\bool $
\end{lemma}

The reverse direction seems not to be provable.

\section{Halting Problems}
\label{sec:synth}
\label{def:K}

For this section we assume $\EA$, i.e.\ $\varphi : \nat \to (\nat \to \option\nat)$ s.t.\ $\forall f : \nat \to \option \nat. \exists c.\;\varphi c \equiv_{\textsf{ran}} f$.
Recall \Cref{coq:EA_to_EA'_prf} stating that
$\forall p : \nat \to \mathbb{P}.~\enumerable p \leftrightarrow \exists c.~\W c \equivwrt{\nat \to \Prop} p$.

We define $\Kintro_0 n \definedas \W n n$ and prove our first negative result:

\setCoqFilename{Axioms.halting}
\begin{lemma}[][K0_enumerable]
  $\neg \enumerable \compl{\K_0}$
\end{lemma}
\begin{proof}
  Assume $\enumerable(\lambda n. \neg \W n n)$.
  By specification of $\W$ there is $c$ s.t.\ $\forall n.\W c n \leftrightarrow \neg \W n n$.
  In particular, $\W c c \leftrightarrow \neg \W c c$, which is contradictory.
\end{proof}
\begin{corollary}[][K0_undec]
  $\neg \decidable \K_0$, $\neg \decidable \compl{\K_0}$, $\neg \decidable \W$ and $\neg \decidable \compl \W$.
\end{corollary}

Intuitively, $\K_0$ can be seen as analogous to the self-halting problem: $\K_0 n$ states that $n$ considered as an enumerator outputs itself in its range (rather than halting on itself).

It is also easy to show that $\W$ and thus $\K_0$ are enumerable:
\begin{lemma}[][enumerable_W]
  $\enumerable \W$
\end{lemma}
\begin{proof}
  Via $f \langle n,m \rangle \definedas \iteis{\varphi n m}{\Some k}{\Some (n,k)}{\None}$.
\end{proof}
\begin{corollary}[][K0_enum]
  $\enumerable \K_0$
\end{corollary}

Since Bauer~\cite{BauerSyntCT} bases his development on $\EA'$, he needs the axiom of countable choice to prove that $\W$ is enumerable, whereas $\EA$ allows an axiom-free proof of this fact.

Another well-known traditional result is that a problem is enumerable if and only if it many-one reduces to the halting problem $\K$, which can be proved without reference to $\EA$.
\begin{align*}
  &p \redm q := \exists f : X \to Y.\forall x.\;p x \leftrightarrow q (f x) &&\K (f : \nat \to \bool) := \exists n.~f n = \btrue
\end{align*}

\setCoqFilename{Synthetic.reductions}
\begin{fact}[][semi_decidable_red_K_iff]
  For all $p : X \to \Prop$, $p \redm \K \leftrightarrow \semidecidable p$.
\end{fact}

\setCoqFilename{Axioms.halting}
\begin{corollary}[][semi_decidable_K]
  $\semidecidable \K$
\end{corollary}
\begin{corollary}[][enumerable_red_K_iff]
  For all $p : \nat \to \Prop$, $p \redm \K \leftrightarrow \enumerable p$.
\end{corollary}

Using the non-enumerability of $\compl{\K_0}$ we can now prove our first negative result by reduction:

\begin{corollary}[][K0_red_K]
  $\K_0 \redm \K$, and thus $\neg \enumerable \compl \K$, $\neg \decidable \K$, and $\neg \decidable \compl \K$.
\end{corollary}

We can also define $\K_\nat \definedas \lambda f : \nat \to \nat.\; \exists n.\; f n \neq 0$:
\begin{fact}[][K_nat_equiv]
  $\K \redm \K_\nat$, $\K_\nat \redm \K$, $\compl{\K_\nat} \equivwrt{(\nat\to\nat) \to \Prop} \lambda f.\; \forall n.\; f n = 0$, and thus\label{coq:forall_PCO_undec}
  $\neg \decidable (\lambda f.\; \forall n.\; f n = 0)$.
\end{fact}

\section{Kleene Trees}
\label{sec:kleene}
\setCoqFilename{Axioms.kleenetree}

In a lecture in 1953 Kleene~\cite{kleene1953recursive} gave an example how the axioms of Brouwer's intuitionism fail if all functions are considered computable by constructing an infinite decidable binary tree with no computable infinite path.
The existence of such a Kleene tree ($\KT$) is in contradiction to Brouwer's fan theorem, which we will discuss later.
We prove that $\EPF_\bool$ implies $\KT$.

For this purpose, we call a predicate $\tau : \List\bool \to \Prop$ a (decidable) \introterm{binary tree} if
\begin{enumerate}[(a)]
\item $\tau$ is decidable: $\exists f. \forall u. \tau u \leftrightarrow f u = \btrue$
\item $\tau$ is non-empty: $\exists u. \tau u$
\item $\tau$ is prefix-closed: If $\tau u_2$ and $u_1 \sqsubseteq u_2$ then $\tau u_1$ (where $u_1 \sqsubseteq u_2 := \exists u'.\;u_2 = u_1 \app u'$).
\end{enumerate}

We will just speak of trees instead of decidable binary trees in the following.

\begin{fact}[][tree_nil]
  For every tree $\tau$, $\tau []$ holds.
\end{fact}

Furthermore, a decidable binary tree $\tau$ \dots
\begin{itemize}
\item \dots is \introterm{bounded} if $\exists n. \forall u. |u| \geq n \to \neg \tau u$
\item \dots is \introterm{well-founded} if $\forall f. \exists n. \neg \tau [f 0, \dots, f n]$
\item \dots has an \introterm{infinite path} if $\exists f. \forall n. \tau [f 0, \dots, f n]$
\end{itemize}

\begin{fact}[][not_bounded_infinite_iff]
  A tree is not bounded if and only if it is \introterm{infinite}, defined as $\forall n. \exists u.\; |u| \geq n \land \tau u$.
\end{fact}

\begin{fact}[][bounded_to_wellfounded]\label{lem:easyfacts_trees}
  Every bounded tree is well-founded and every tree with an infinite path is infinite.
\end{fact}

Note that both implications are strict:
In our setting we cannot prove boundedness from well-foundedness nor obtain an infinite path from infiniteness, as can be~seen~from~a~Kleene~tree:%
\label{def:KT}%
\[\KTintro \definedas \textit{There exists an infinite, well-founded, decidable binary tree.}\]
We follow Bauer~\cite{bauer2006konig} to construct a Kleene tree.

\begin{lemma}[][diag]
  Given $\EPF_\bool$ one can construct $d : \nat \pfun \bool$ s.t.\ $\forall f : \nat \to \bool.\exists n b.\; d n \hasvalue b \land f n \neq b$.
\end{lemma}
\begin{proof}
  Define $d n := e n n \bind \lambda b.\; \ret (\neg_\bool b) $.
\end{proof}

We define $\tau_K u \definedas \forall n < \length u.\forall x.\;\textsf{seval}~ (d n) ~ \length u = \Some x \to u[n] = \Some x$.
Intuitively, $\tau_K$ contains all paths $u = [b_0, b_1, \dots, b_n]$ which might be prefixes of $d$ given $n$ as step index, i.e.\ where $n$ does not suffice to verify that $d$ is no prefix of $d$.
An infinite path through $\tau_K$ would be a totalisation of $d$.

\begin{theorem}[][T_K]
  $\EPF_\bool \to \KT$
\end{theorem}
\begin{proof}
  We show that $\tau_K$ is a Kleene tree.
  That $\tau_K$ is a decidable tree is immediate.
  To show that $\tau_K$ is infinite let $k$ be given. We define $f 0 \definedas []$ and $f (S n) \definedas f n \app [\iteis{D k n}{\Some x}{x}{\bfalse}]$.
  We have $\length{f n} = n$.
  In particular, $\length{f k} \geq k$ and $\tau_K (f k)$.

  For well-foundedness let $f : \nat \to \bool$ be given.
  There is $n$ s.t.\ $d n \hasvalue b$ and $f n \neq b$.
  Thus there is $k$ s.t.\ \mbox{$\textsf{seval} ~ (d n) ~ k = \Some b$}.
  Now $\neg \tau_K u$ for $u \definedas [f 0, \dots, f (n + k)]$. %
\end{proof}

\section{Extensionality Axioms}
\label{sec:ext}
\setCoqFilename{Axioms.principles}

Coq's type theory is intensional, i.e.\ $f \equivwrt{A \to B} g$ and $f = g$ do not coincide.
Extensionality properties can however be consistently assumed as axioms.
In this section we briefly discuss the relationship between $\CT$ and functional extensionality $\Fext$, propositional extensionality $\Pext$ and proof irrelevance $\PI$, defined as follows:
\label{def:Fext}\label{def:Pext}\label{def:PI}
\begin{align*}
  \Fextintro &\definedas \forall A B.\forall f g : A \to B.\; (\forall a.f a = g a) \to f = g \\
  \Pextintro &\definedas \forall P Q : \Prop.\; (P \leftrightarrow Q) \to P = Q \\
  \PIintro &\definedas \forall P : \Prop. \forall (x_1 x_2 : P).\; x_1 = x_2
\end{align*}

\begin{fact}[][Pext_to_PI]
  $\Pext \to \PI$
\end{fact}

Swan and Uemura~\cite{swan2019church} prove that intensional predicative Martin-Löf type theory remains consistent if $\CT$, the axiom of univalence, and propositional truncation are added.
Since functional extensionality and propositional extensionality are a consequence of univalence, and propositions are semantically defined as exactly the irrelevant types, $\Fext$, $\Pext$, and $\PI$ hold in this extension of type theory.
It seems likely that the consistency result can then be adapted to Coq's type theory, yielding a consistency proof for $\CT$ with $\Fext$, $\Pext$, and $\PI$.

It is however crucial to formulate $\CT$ using $\exists$ instead of $\Sigma$.
The formulation as $\CT_\Sigma := \forall f.~\Sigma n.~n\sim f$ is inconsistent with functional extensionality $\Fext$, as already observed in~\cite{troelstra1988constructivism}.

\begin{lemma}[][CT_Sigma_wrong]
  $\CT_\Sigma \to \Fext \to \bot$
\end{lemma}
\begin{proof}
  Since $\CT_\Sigma$ implies $\EA$, it suffices to prove that $\lambda f. \forall n.\; f n = 0$ is decidable by \Cref{coq:forall_PCO_undec}.
  Assume $G : \forall f.~\Sigma c.~c\sim f$ and let $F f := \ite{\pi_1 (G f) = \pi_1 (G (\lambda x.0))}{\btrue}{\bfalse}$.
  
  If $F f = \btrue$, then $\pi_1(G f) = \pi_1 ( G (\lambda x.0))$ and by extensionality of $\sim$, $f n = (\lambda x.0) n = 0$.

  If $\forall n.~f n = 0$, then $f = \lambda x.\; 0$ by $\Fext$, thus $\pi_1(G f) = \pi_1 ( G (\lambda x.\; 0))$ and $F f = \btrue$.  
\end{proof}

\newcommand\sat[1]{(\exists n.~#1 n = \btrue)}
\section{Classical Logical Axioms}
\label{sec:class}

In this section we consider consequences of the law of excluded middle $\LEM$.\label{def:DNE}
Precisely, besides $\LEM$, we consider the weak law of excluded middle $\WLEM$, the Gödel-Dummett-Principle~$\DGP$\rlap,\footnote{We follow Diener~\cite{dienerConstructiveReverseMathematics2020} in using the abbreviation $\DGP$ instead of $\mathsf{GDP}$.} and the principle of independence of premises $\IP$, together with their respective restriction of propositions to the satisfiability of boolean functions, resulting in the limited principle of omniscience $\LPO$, the weak limited principle of omniscience $\WLPO$, and the lesser limited principle of omniscience $\LLPO$.

\label{def:LEM}\label{def:WLEM}\label{def:DGP}
\label{def:LPO}\label{def:WLPO}\label{def:LLPO}\label{def:IP}
\lapbox{-0.8cm}{\scalebox{0.88}{\parbox{\linewidth}{%
\begin{align*}
  \LEMintro &:= \forall P : \Prop.~P \lor \neg P & \LPO &:= \forall f : \nat \to \bool.\; \sat f \lor \neg \sat f \\
  \WLEMintro &:= \forall P : \Prop.~\neg\neg P \lor \neg P &\WLPO &:= \forall f : \nat \to \bool.~\neg\neg\sat f \lor \neg \sat f \\
  \DGPintro &:= \forall P Q : \Prop. (P \to Q) \lor (Q \to P) &\LLPO &:= \forall f g : \nat \to \bool.\; (\sat f \to \sat g)  \\
  &&&\hspace{2.2cm}~ \lor (\sat g \to \sat f) \\  %
  \IPintro &:= &&\hspace{-5.15cm}\forall P : \Prop.\forall q : \nat \to \Prop.\; (P \to \exists n. q n) \to \exists n.\; P \to q n 
\end{align*}
}}}
\vspace{-0.6\baselineskip}
\begin{fact}[][LEM_to_DGP]
  $\LEM \to \DGP$, $\DGP \to \WLEM$, $\LEM \to \IP$.
\end{fact}

The converses are likely not provable:
Diener constructs a topological model where $\DGP$ holds but not $\LEM$, and one where $\WLEM$ holds but not $\DGP$~\cite[Proposition 8.5.3]{dienerConstructiveReverseMathematics2020}.
Pédrot and Tabareau~\cite{PedrotMP} construct a syntactic model where $\IP$ holds, but $\LEM$ does not.

\begin{fact}[][LPO_to_WLPO]
  $\LPO \to \WLPO$ and $\WLPO \to \LLPO$.
\end{fact}

The converses are likely not provable:
Both implications are strict over $\mathsf{IZF}$ with dependent choice \cite[Theorem 5.1]{hendtlassSEPARATINGFRAGMENTSWLEM2016}.%

$\LPO$ is $\Sigma^0_1\text{-}\LEM$ and $\WLPO$ is simultaneously $\Sigma^0_1\text{-}\WLEM$ and $\Pi^0_1\text{-}\LEM$, due to the following:

\begin{fact}[][forall_neg_exists]\label{fact:forall_neg}
  $(\forall n. f n = \bfalse) \leftrightarrow \neg (\exists n. f n = \btrue)$
\end{fact}

Both can also be formulated for predicates:

\begin{fact}[][LPO_semidecidable_iff]\label{lem:LPO_equivs}
  The following equivalences hold:
  \begin{enumerate}
  \coqitem[LPO_semidecidable_iff] $\LPO \phantom{\textsf{W}}\leftrightarrow \forall X.\forall (p : X \to \Prop).~\semidecidable  p \to \forall x. \phantom{\neg}p x \lor \neg p x$
  \coqitem[WLPO_semidecidable_iff] $\WLPO \leftrightarrow \forall X.\forall (p : X \to \Prop).~\semidecidable  p \to \forall x. \neg p x \lor \neg \neg p x$
  \coqitem[WLPO_cosemidecidable_iff] $\WLPO \leftrightarrow \forall X.\forall (p : X \to \Prop).~\compl\semidecidable  p \to \forall x. \phantom{\neg}p x \lor \neg p x$
  \end{enumerate}
\end{fact}

In our formulation, $\LLPO$ is the Gödel-Dummet rule for $\Sigma^0_1$ propositions.
It can also be formulated as $\Sigma^0_1$ or $\semidecidable$ De Morgan rule (\lipicsNumber{2}, \lipicsNumber{3} in the following Lemma), $\semidecidable$-\DGP (\lipicsNumber{4}), or as a double negation elimination principle on $\compl\semidecidable$ relations into booleans (\lipicsNumber{5}):

\begin{lemma}[][LLPO_to_DM_Sigma_0_1]\label{lem:LLPO_equivs}
  The following are equivalent:
  \begin{enumerate}
    \coqitem[LLPO_to_DM_Sigma_0_1] $\LLPO$
    \coqitem[DM_Sigma_0_1_to_LLPO_split] \small $\forall f g : \nat \to \bool.\; \neg((\exists n. f n = \btrue) \land (\exists n. g n = \btrue)) \to \neg (\exists n. f n = \btrue) \lor \neg (\exists n. g n = \btrue)$
    \coqitem[DM_Sigma_0_1_iff_DM_sdec] $\forall X.\forall (p~q : X \to \Prop).\;\semidecidable  p \to \semidecidable q \to \forall x.\; \neg( p x \land  q x) \to \neg p x \lor \neg q x$
    \coqitem[LLPO_iff_DGP_sdec] $\forall X.\forall (p : X \to \Prop).\;\semidecidable  p \to \forall x y.\; (p x \to p y) \lor (p y \to p x)$
    \coqitem[DM_Sigma_0_1_iff_totality] $\forall X.\forall (R : X \to \bool \to \Prop).\; \compl\semidecidable R \to \forall x.\; \neg\neg (\exists b.~R x b) \to \exists b.~R x b$
    \coqitem[LLPO_split_to_LLPO] \small $\forall f.\; (\forall n m. f n = \btrue \to f m = \btrue \to n = m) \to (\forall n. f (2 n) = \bfalse) \lor (\forall n. f (2 n + 1) = \bfalse)$
  \end{enumerate}
\end{lemma}
\label{def:PFP}

We define the \introterm{principle of finite possibility} as $\PFPintro \definedas \forall f.\exists g.\; (\forall n.\; f n = \bfalse) \leftrightarrow (\exists n.\; g n = \btrue)$. $\PFP$ unifies $\WLPO$ and $\LLPO$:
\begin{fact}[][WLPO_PFP_LLPO_iff]
  $\WLPO \leftrightarrow \LLPO \land \PFP$
\end{fact}

\sfcommand{KS}\label{def:KS}
A principle unifying the classical axioms with their counterparts for $\Sigma^0_1$ is \introterm{Kripke's schema} $\KSintro \definedas \forall P : \Prop.\exists f : \nat \to \bool.~P \leftrightarrow \exists n.\; f n = \btrue$:
\begin{fact}[][LEM_to_KS]
  $\LEM \to \KS$
\end{fact}
\begin{fact}[][KS_LPO_to_LEM]
  Given $\KS$ we have $\LPO \to \LEM$, $\WLPO \to \WLEM$, and $\LLPO \to \DGP$.
\end{fact}

$\KS$ could be strengthened to state that every predicate is semi-decidable (to which $\KS$ is equivalent using $\AC_{\nat,\nat\to\nat}$).
The strengthening would be incompatible with $\CT$.%

In general, the compatibility of classical logical axioms (without assuming choice principles) with $\CT$ seems open.
We conjecture that Coq's restriction preventing large elimination principles for non-sub-singleton propositions makes $\LEM$ and $\CT$ consistent in Coq.

\section{Axioms of Russian Constructivism}
\label{sec:russ}
\label{def:MP}

The Russian school of constructivism morally identifies functions with computable functions, sometimes assuming $\CT$ explicitly.
Another axiom considered valid is Markov's principle:
\[ \MP \definedas \forall f : \nat \to \bool.\;\neg\neg (\exists n.\;f n = \btrue) \to \exists n.\;f n = \btrue\]

Markov's principle is consistent with \CT \cite{swan2019church} and follows from $\LPO$:

\begin{fact}[][LPO_MP_WLPO_iff]
  $\LPO \leftrightarrow \WLPO \land \MP$
\end{fact}

\begin{corollary}[][LPO_to_MP]
  $\LPO \to \MP$.
\end{corollary}

It seems likely that the converse is not provable:
There is a logic where $\MP$ holds, but not $\LPO$ \cite{herbelin2010intuitionistic}.
As observed by Herbelin~\cite{herbelin2010intuitionistic} and Pedrót and Tabareau~\cite{PedrotMP}, $\IP \land \MP$ yields $\LPO$:

\begin{lemma}[][MP_IP_LPO]
  $\MP \to \IP \to \LPO$
\end{lemma}
\begin{proof}
  Given $f : \nat \to \bool$ there is $n_0 : \nat$ s.t.\ $\forall k.\; f k = \btrue \to f n_0 = \btrue$ using $\MP$ and $\IP$:
  By $\MP$, $\neg\neg (\exists k.\; f k = \btrue) \to \exists n.\; f n = \btrue$ and by $\IP$, $\exists n. \neg\neg(\exists k. f k = \btrue) \to f n = \btrue$, which suffices.
  Now $f n_0 = \btrue \leftrightarrow \exists n.\; f n = \btrue$ and $\LPO$ follows.
\end{proof}

A nicer factorisation would be to prove $\IP \to \WLPO$, but the implication seems unlikely.

\begin{lemma}[][MP_to_MP_semidecidable] The following are equivalent:
  \begin{enumerate}
  \coqitem[MP_to_MP_semidecidable] $\MP$
  \coqitem[MP_semidecidable_to_Post_logical]  $\forall X.\forall p : X \to \Prop.\;\semidecidable p \to \forall x.\;\neg\neg p x \to p x$
  \coqitem[Post_logical_to_Post] $\forall X.\forall p : X \to \Prop.\;\semidecidable p \to \semidecidable \compl p \to \forall x.\;p x \lor \neg p x$
  \coqitem[Post_to_MP] $\forall X.\forall p : X \to \Prop.\;\semidecidable p \to \semidecidable \compl p \to \decidable p$
  \coqitem[MP_iff_sdec_weak_total] $\forall X.\forall (R : X \to \bool \to \Prop).\; \semidecidable R \to \forall x.\; \neg\neg (\exists b.~R x b) \to \exists b.~R x b$
  \end{enumerate}
\end{lemma}
\begin{proof}
  \begin{itemize}
  \item $\lipicsNumber{1} \to \lipicsNumber 2$ is immediate.
  \item $\lipicsNumber 2 \to \lipicsNumber 3$:
    Since $\semidecidable$ is closed under disjunctions
    and since $\neg\neg(p x \lor \neg p x)$ is a tautology.
  \item $\lipicsNumber 3 \to \lipicsNumber 4$ is immediate by \Cref{lem:andrej_sdec} with $R x b := (p x \land b = \btrue) \lor (\neg p x \land b = \bfalse)$.
  \item $\lipicsNumber 4 \to \lipicsNumber 1$:
    Let $\neg\neg (\exists n. f n = \btrue)$.
    Let $p (x : \nat) \definedas  \exists n. f n = \btrue$.
    Now $p$ is semi-decided by $\lambda x.f$, $\compl p$ by $\lambda x n. \bfalse$, and $p 0 \lor \neg p 0$ by $\lipicsNumber{4}$. One case is easy, the other contradictory.
    \hfill\popQED
  \end{itemize}
  
\end{proof}

Note that \lipicsNumber{4} is often called ``Post's theorem''.
$\lipicsNumber{1} \leftrightarrow \lipicsNumber{3} \leftrightarrow \lipicsNumber{4}$ is already discussed in~\cite{forster2019synthetic}.
\lipicsNumber{5} is dual to \Cref{lem:LLPO_equivs} (\lipicsNumber{5}).
Replacing $\semidecidable p$ with $\compl\semidecidable p$ in \lipicsNumber{2} does however not result in an equivalent of $\LLPO$, but turns \lipicsNumber{2} into an assumption-free fact.
While in general $\semidecidable{\compl p} \leftrightarrow \compl\semidecidable p$ does not hold it seems possible that they can be exchanged in $\lipicsNumber{3}$ and $\lipicsNumber{4}$, but we are not aware of a proof.

\section{Choice Axioms}
\label{sec:choice}

We consider the axioms of choice $\AC$, unique choice $\AUC$, dependent choice $\ADC$, and countable choice $\ACC$. $\AC_{\nat,\nat}$ and $\AC_{\nat\to\nat,\nat}$ are often called $\AC_{0,0}$ and $\AC_{1,0}$ in the literature.%
{\small
  \label{def:AC}\label{def:AUC}\label{def:ADC}\label{def:ACC}
\begin{align*}
  \ACintro_{X,Y} &\definedas \forall R : X \to Y \to \Prop. (\forall x. \exists y.R x y) \to \exists f : X \to Y. \forall x.\; R x (f x) \\
  \AUCintro_{X,Y} &\definedas \forall R : X \to Y \to \Prop. (\forall x. \exists ! y.R x y) \to \exists f : X \to Y. \forall x.\; R x (f x) \\
  \ADCintro_X &\definedas \forall R : X \to X \to \Prop.(\forall x. \exists x'.R x x') \to \forall x_0. \exists f : \nat \to X. f 0 = x_0 \land \forall n.\; R (f n) (f (n + 1)))
\end{align*}%
\vspace{-1.8\baselineskip}
\begin{align*}
  &\ACintro := \forall X Y : \Type.\; \AC_{X,Y} && \AUCintro := \forall X Y.\; \AUC_{X,Y} & \ADCintro := \forall X : \Type.\; \ADC_{X} &&\ACCintro := \forall X : \Type.\; \AC_{\nat,X}  
\end{align*}}
\vspace{-1\baselineskip}
\begin{fact}[][AC_rel_to_ADC]\footnotesize{\footnotesize
  $\AC_{X,X} \to \ADC_X$, $\AC_{X,Y} \to \AUC_{X,Y}$, $\ADC \to \ACC$, $\ACC \to \AC_{\nat,\nat}$}, and {\footnotesize~$\AC_{\nat\to\nat,\nat} \to \AC_{\nat,\nat}$}.
\end{fact}

The following well-known fact is due to Diaconescu~\cite{diaconescu1975axiom} and Myhill and Goodman~\cite{doi:10.1002/malq.19780242514}:

\begin{fact}[][Diaconescu]
  $\AC \to \Fext \to \Pext \to \LEM$
\end{fact}

Given that $\AC_{\nat\to\nat,\nat}$ turns $\CT$ into $\CT_\Sigma$, and that $\EA \leftrightarrow \Sigma T. \CT_T$ we have:%

\begin{fact}[][AC_1_0_Fext_incompat]
  $\AC_{\nat\to\nat,\nat} \to \Fext \to \EA \to \bot$
\end{fact}

We will later see that $\LLPO \land \AC_{\nat,\nat}$ implies weak Kőnig's lemma, which is incompatible with $\KT$.
Already now we can prove that $\WLPO \land \AUC_{\nat,\bool}$ is incompatible with $\EA$:

\begin{fact}[][AUC_to_dec]
  $\AUC_{\nat,\bool} \to (\forall n : \nat.~p n \lor \neg p n) \to \decidable p$
\end{fact}

\begin{lemma}[][AC_0_0_LPO_incompat']
  $\WLPO \to \AUC_{\nat, \bool} \to \EA \to \decidable \compl \K_0$
\end{lemma}
\begin{proof}
  $\WLPO$ implies $\forall n. \neg \K_0 n \lor \neg\neg \K_0 n$.
  By $\AUC_{\nat,\bool}$ and the last lemma $\compl \K_0$ is decidable.
\end{proof}
\begin{corollary}[][AC_0_0_LPO_incompat]
  $\WLPO \to \AUC_{\nat, \bool} \to \EA \to \bot$
\end{corollary}

\subsection{Provable choice axioms}

In contrast to predicative Martin-Löf type theory, Coq's type theory does not prove the axiom of choice, nor the axioms of dependent and countable choice.
This is due to the fact that arbitrary large eliminations are not allowed.
However, recall that a large elimination principle for the accessibility predicate is provable, resulting in \Cref{coq:mu_nat}.
Using \Cref{coq:mu_nat} we can then prove $\decidable\textsf{-}\AC_{X,\nat}$ for all $X$, i.e.\ choice for decidable relations into natural numbers:
\setCoqFilename{Axioms.principles}

\begin{lemma}[][decidable_AC]
  $\forall X. \forall R : X \to \nat \to \Prop.~\decidable R \to (\forall x.\exists n. R x n) \to \exists f : X \to \nat.\forall x.\;R x (f x)$.  
\end{lemma}

As a consequence and with no further reference to \Cref{coq:mu_nat} we can then prove choice principles for semi-decidable and enumerable relations, i.e.\ $\semidecidable\textsf{-}\AC_{X,\nat}$ and $\enumerable\textsf{-}\AC_{\nat,X}$ for all $X$:

\begin{lemma}[][semi_decidable_AC]\label{lem:andrej_sdec}
  The following two choice principles are provable\footnote{A formulation of (\lipicsNumber{1}) for disjunctions (equivalently: $R : X \to \bool \to \Prop$) is due to Andrej Dudenhefner and was received in private communication. ($\lipicsNumber{2}$) was anticipated by~Larchey-Wendling~\cite{larchey2017typing}, who formulated it for $\mu$-recursively enumerable instead of synthetically enumerable predicates.}:
  \begin{enumerate}
    \coqitem[semi_decidable_AC]
    $\forall X.\forall R : X \to \nat \to \Prop.~ \semidecidable R \to (\forall x.\exists n.\; R x n) \to \exists f : X \to \nat.\forall x.\;R x (f x)$
  \coqitem[enumerable_AC]
  $\forall X.\forall R : \nat \to X \to \Prop.~\enumerable R \to (\forall n.\exists x.\;R n x) \to \exists f : \nat \to X.\forall n.\;R n (f n)$
  \end{enumerate}
\end{lemma}

Principle $\lipicsNumber{2}$ can be relaxed to arbitrary discrete types instead of $\nat$, and in particular $\semidecidable\textsf{-}\AC_{\nat,\bool}$ follows from $\lipicsNumber{1}$.
In \Cref{sec:modesty} we discuss consequences of the here mentioned principles with regards to $\CT$~for~oracles and
in the next section $\compl\semidecidable\textsf{-}\AC_{\nat,\bool}$ will~be~central. %

\section{Axioms on Trees}
\label{sec:trees}
\setCoqFilename{Axioms.kleenetree}

We have already introduced (decidable) binary trees and Kleene trees in \Cref{sec:kleene}.
We now give a broader overview and give formulations of $\LPO$, $\WLPO$, $\LLPO$, and $\MP$ in terms of decidable binary trees, following Berger et al.~\cite{berger2012weak}.

\begin{fact}[][is_tree_subtree_at]
  Let $\tau$ be a tree. Then $\tau_u v \definedas \tau (u \app v)$ is a tree if and only if $\tau u$.
\end{fact}

If $\tau u$ holds we call $\tau_u$ a \introterm{subtree} of $\tau$ and $\tau_{[b]}$ a \introterm{direct~subtree} of $\tau$.%

\begin{lemma}[][LPO_tree_iff]\label{lem:omni_equivs}
  The following equivalences hold:
  \begin{enumerate}
    \coqitem[LPO_tree_iff] $\LPO \leftrightarrow$ every tree is bounded or infinite.
    \coqitem[WLPO_tree_iff] $\WLPO \leftrightarrow$  every tree is infinite or not infinite.
    \coqitem[LLPO_tree_iff] $\LLPO \leftrightarrow$ every infinite tree has a direct infinite subtree.
    \coqitem[MP_tree_iff] $\MP \leftrightarrow$ if a tree is not infinite it is bounded.
    \coqitem[MP_tree_iff'] $\MP \leftrightarrow$ if a tree has no infinite path it is well-founded.
  \end{enumerate}
\end{lemma}

Recall \Cref{lem:easyfacts_trees} stating that every bounded tree is well-founded and that every tree with an infinite path is infinite.
The respective converse implications are known as Brouwer's \introterm{fan theorem} \FAN and~\introterm{weak Kőnig's lemma} \WKL respectively:
\label{def:FAN}\label{def:WKL}
\begin{align*}
  \FAN &\definedas \textit{Every well-founded decidable binary tree is bounded.} \\
  \WKL &\definedas \textit{Every infinite decidable binary tree has an infinite path.} 
\end{align*}

\begin{fact}[][KT_FAN_contra]
  $\KT \to \neg \FAN$ and $\KT \to \neg \WKL$.
\end{fact}

Note that $\FAN$ is called $\FAN'_\Delta$ in~\cite{ishihara2006weak} and $\FAN_\Delta$ in \cite{dienerConstructiveReverseMathematics2020}, and $\WKL$ is called $\WKL_{\decidable{}}$ in~\cite{forster2020completenessext}. %
Ishihara~\cite{ishihara2006weak} shows how  to deduce $\FAN$ from $\WKL$ constructively:

\begin{fact}[][bounded_longest_path]
  Bounded trees $\tau$ have a longest element, i.e.\ $\exists u. \;\tau u \land \forall v.~ \tau v \to \length v \leq \length u$.
\end{fact}

\begin{lemma}[][inf_to_longest]
  For every tree $\tau$ there is an infinite tree $\tau'$ s.t.\ for any infinite path $f$ of $\tau'$ 
  $\forall u.\; \tau u \to \tau [f 0, \dots, f |u|]$.
\end{lemma}

\begin{theorem}[][WKL_to_FAN]
  $\WKL \to \FAN$ 
\end{theorem}
\begin{proof}
  Let $\tau$ be well-founded.
  By \Cref{coq:inf_to_longest} and $\WKL$, there is $f$ s.t.\ $\forall a.~ \tau u \to \tau [f 0, \dots, f \length u]$.
  Since $\tau$ is well-founded there is $n$ s.t.\ $\neg \tau [f 0, \dots, f n]$.
  Then $n$ is a bound for $\tau$: For $u$ with $\length u > n$ and $\tau u$ we have $\tau [f 0, \dots, f n, \dots , f \length u]$.
  But then $\tau [f 0, \dots, f n]$, contradiction.
\end{proof}
\begin{corollary}[][KT_WKL_contra]
  $\KT \to \neg \WKL$.
\end{corollary}

Berger and Ishihara~\cite{berger2005brouwer} show that $\FAN \leftrightarrow \WKL!$, a restriction of \WKL stating that every infinite decidable binary tree with \textit{at most one} infinite path has an infinite path.
Schwichtenberg~\cite{schwichtenberg2005direct} gives a more direct construction and mechanises the proof in Minlog.

Berger, Ishihara, and Schuster~\cite{berger2012weak} characterise $\WKL$ as the combination of the logical principle $\LLPO$ and the function existence principle~$\compl\semidecidable\textsf{-}\AC_{\nat,\bool}$ (called $\Pi^0_1\CC^\lor$ in \cite{berger2012weak}).
We observe that $\WKL$ can also be characterised as one particular choice or dependent choice principle.
The proofs are essentially rearrangements of \cite[Theorem 27 and Corollary 5]{berger2012weak}.

\begin{theorem}[][WKL_to_LLPO]\label{thm:WKL_equivs}
  The following are equivalent:
  \begin{enumerate}
    \coqitem[WKL_to_LLPO] $\WKL $
    \coqitem[LLPO_coS_AC_on_to_coS_AC_on_weak] $\LLPO \land \compl\semidecidable\textsf{-}\AC_{\nat,\bool}$
    \coqitem[cos_AC_on_weak_to_coS_ADC_on_weak] {$\forall R : \phantom{\List}\nat \to \bool \to \Prop.~\compl\semidecidable R \to (\forall n.\neg\neg\exists b. R n b) \to \exists f : \nat \to \bool.\forall n.~R~n~(f n)$}
    \coqitem[coS_ADC_on_weak_to_WKL] { $\forall R : \List \bool \to \bool \to \Prop.~\compl\semidecidable R \to (\forall u.\neg\neg\exists b. R u b) \to \exists f : \nat \to \bool.\forall n.~R~[f 0, \dots, f (n - 1)]~(f n)$}
  \end{enumerate}
\end{theorem}
\begin{proof}
  For $\WKL \to \LLPO$ we use the characterisation $\lipicsNumber{3}$ of $\LLPO$ from \Cref{lem:omni_equivs}.
  Let $\tau$ be an infinite tree. By $\WKL$ there is an infinite path $f$.
  Then $\tau_{[f 0]}$ is a direct infinite subtree.

  For $\WKL \to \compl\semidecidable\textsf{-}\AC_{\nat,\bool}$ let $R$ be total and $f$ s.t.\ $\forall n b.~R n b \leftrightarrow \forall m. f n b m = \bfalse$.
  Define the tree $\tau u := \forall i < \length u. \forall m < \length u.~f i (u[i]) m = \bfalse$.
  Infinity of $\tau$ follows from $\forall n.\exists u.\length u = n \land \forall i < n.R i (u[i])$, proved by induction on $n$ using totality of $R$.
  If $g$ is an infinite path of $\tau$, $R n (g n)$ follows from $\forall m.\tau [g 0, \dots, g (n + m + 1)]$.

  $\lipicsNumber{2} \to \lipicsNumber{3}$ is immediate using characterisation~$\lipicsNumber{3}$ of $\LLPO$ from \Cref{lem:LLPO_equivs}.

  For $\lipicsNumber{3} \to \lipicsNumber{4}$ let $F : \nat \to \List\bool$ and $G : \List \bool \to \nat$ invert each other\rlap.\footnote{These so called coding functions is easy to construct even formally using e.g.\ techniques from~\cite{forster2019synthetic}.}
  Let $R : \List \bool \to \bool \to \Prop$ and  $f$ be the choice function obtained from $\lipicsNumber{3}$ for $\lambda n b.R (F n) b$.
  Then $\lambda n.f (G (g n))$ where $g 0 := []$ and $g (\succN n) := g n \app [f (G (g n))]$ is a choice function for $R$ as wanted.
  
  For $\lipicsNumber{4} \to \lipicsNumber{1}$ let $\tau$ be an infinite tree and let $d_u m := \exists v. \length v = m \land \tau_u v$, i.e.\ $d_u m$ if $\tau_u$ has depth at least $m$ and in particular $\tau_u$ is infinite iff $\forall m.d_u m$.
  Define $R u b := \forall m. d_{u \app [b]} m \lor \neg d_{u \app [\neg_\bool b]}$.
  $R$ is co-semi-decidable (since $d$ is decidable), and $\neg R\,u\, \btrue \land \neg R\, u\, \bfalse$ is contradictory.
  Thus $\lipicsNumber{4}$ yields a choice function $f$ which fulfils $\tau [f0, \dots, f n]$ by induction on $n$.
\end{proof}

\section{Continuity:\hspace{0.2em}Baire Space,\hspace{0.2em}Cantor Space,\hspace{0.2em}and Brouwer's Intuitionism}
\label{sec:cont}
\setCoqFilename{Axioms.baire_cantor}

The total function space $\nat \to \nat$ is often called \introterm{Baire space}, whereas $\nat \to \bool$ is called \introterm{Cantor space}.
We will from now on write $\nat^\nat$ and $\bool^\nat$ for the spaces.

Constructively, one cannot prove that $\nat^\nat$ and $\bool^\nat$ are in bijection.
However, $\KT$ is equivalent to the existence of a continuous bijection $\bool^\nat \to \nat^\nat$ with a continuous modulus of continuity, i.e.\ a modulus function which is continuous (in the point) itself~\cite{dienerConstructiveReverseMathematics2020}.
Furthermore, %
$\KT$ yields a continuous bijection $\nat^\nat \to \bool^\nat$~\cite{beeson2012foundations}.%

We call a function $F : A^\nat \to B^\nat$ \introterm{continuous} if $\forall f : A^\nat.\forall n : \nat.\exists L : \List\nat. \forall g : A^\nat.\; (\map~f~L = \map~g~L) \to F f n = F g n$.
A function $M : A^\nat \to \nat \to \List \nat$ is called the \introterm{modulus of continuity} for $F$ if
$\forall n : \nat.\forall f g : A^\nat.~\textsf{map}~f~(M f n) = \textsf{map}~g~(M f n) \to F f n = F g n$.
We define:
\label{def:Homeo}
\begin{align*}
  \Homeointro(A^\nat,B^\nat) &:= \exists F: A^\nat \to B^\nat.\exists M.~ M\,\textit{is a continuous modulus of continuity for } F
\end{align*}

We start by proving that $\KT \leftrightarrow \Homeo(\bool^\nat,\nat^\nat)$.
To do so, we say that $u \app [b]$ is a \introterm{leaf of a Kleene tree} $\tau_K$ if $\tau_K u$, but $\neg \tau_K (u \app [b])$.

\begin{fact}[][KT_inj_enum_leafs]
  For every $\tau_K$, there is an injective enumeration $\ell : \nat \to \List \bool$ of the leaves of $\tau_K$.
\end{fact}

We define $F (f : \nat \to \nat) n := (\ell(f 0) \app \dots \app \ell(f (n + 1)))[n]$.
Since leaves cannot be empty, the length of the accessed list is always larger than $n$ and $F$ is well-defined.

\begin{lemma}[][F_inj]
  $F$ is injective w.r.t. $\equivwrt{\nat^\bool}$ and $\equivwrt{\nat^\nat}$.
\end{lemma}

\begin{lemma}[][continuous_F]\label{lem:F_cont}
  $F$ is continuous with continuous modulus of continuity.
\end{lemma}

\sfcommandsp{pref}\label{def:pref}
\sfcommandsp{ind}\label{def:ind}

\begin{lemma}[][F_find_pref] The following hold for a Kleene tree $\tau_K$:
  \begin{enumerate}
  \item There is a function $\ell^{-1} : \List \bool \to \nat$ s.t.\ for all leafs $l$, $\ell (\ell^{-1} l) = l$.
  \item For all $l$ s.t.\ $\neg \tau_K l$ there exists $l' \sqsubseteq l$ s.t.\ $l'$ is a leaf of $\tau_K$.
  \item There is $\pref : (\nat \to \bool) \to \List \bool$ s.t.\ $\pref g$ is a leaf of $\tau_K$ and $\exists n.~\pref g = \map~g~[0, \dots, n]$.
  \end{enumerate}
\end{lemma}

\sfcommand{nxt}\label{def:nxt}%
We can now define the inverse as $G~g~n\definedas \ell^{-1}(\pref (\nxt^n g))$ where $\nxt~g~n \definedas g (n + \length{ \pref g  })$.%
\begin{lemma}[][F_surj]\label{lem:F_inv}
  $F~(G~g) \equivwrt{\nat \to \bool} g$
\end{lemma}

\begin{lemma}[][continuous_G]\label{lem:G_cont}
  $G$ is continuous with continuous modulus of continuity.
\end{lemma}

The following proof is due to Diener~\cite[Proposition 5.3.2]{dienerConstructiveReverseMathematics2020}.

\sfcommand{extend}

\begin{lemma}[][Homeo_M_Cantor_Baire_to_KT]
  $\Homeo(\bool^\nat,\nat^\bool) \to \KT$
\end{lemma}
\begin{proof}
  Let $F$ be a bijection with continuous modulus of continuity $M$.
  Then $\tau u \definedas \forall 0 < i \leq |u|.\exists k < i. k \in M (\lambda n. \iteis{l[n]}{\Some b}{b}{\bfalse})~0$ is a Kleene tree.
\end{proof}
\begin{theorem}[][KT_iff_Homeo_N_nat_bool]
  $\KT \leftrightarrow \Homeo(\bool^\nat,\nat^\nat)$ and $\KT \to \Homeo(\nat^\nat,\bool^\nat)$.
\end{theorem}

Deiser~\cite{deiser2009simple} proves in a classical setting that $\Homeo(\nat^\nat,\bool^\nat)$ holds.
It would be interesting to see whether the proof can be adapted to a constructive proof $\WKL \to \Homeo(\nat^\nat,\bool^\nat)$.

We have already seen that $\CT$ is inconsistent with $\FAN$.
Besides $\FAN$, in Brouwer's intuitionism the continuity of functionals $\nat^\nat \to \nat$ is routinely assumed:
\label{def:Cont}
\[\text{\hspace{-0.5cm}
\scalebox{0.90}{
\[ \Cont \definedas \forall F : (\nat \to \nat) \to \nat.~ \forall f : \nat \to A.\exists L : \List\nat. \forall g : \nat \to A.\; (\map~f~L = \map~g~L) \to F f \equivwrt{B} F g \]}}\]

Since every computable function is continuous, we believe $\Cont$ to be consistent with $\CT$.
Combining $\Cont$ with $\AC_{\nat\to\nat,\nat}$ yields \introterm{Brouwer's continuity principle}\rlap,\footnote{But note that $\Cont \to \AC_{\nat\to\nat,\nat} \to \bot$, since the resulting modulus of continuity function allows for the construction of a non-continuous function~\cite{hotzel2015inconsistency}.} called $\WCN$ in~\cite{troelstra1988constructivism}:
\label{def:WCN}
\[ \WCN \definedas \forall R : (\nat \to \nat) \to \nat \to \Prop.(\forall f. \exists n. R f n) \to \forall f.\exists L n. \forall g.\; \map~f~L = \map~g~L \to R g n \]

\setCoqFilename{Axioms.principles}
\begin{theorem}[][WC_N_to_Cont]
  $\WCN \to \Cont$
\end{theorem}

$\WCN$ is inconsistent with $\CT$, since the computability relation $\sim$ is not continuous:

\begin{theorem}[][WC_N_CT_inc]
  $\WCN \to \CT \to \bot$
\end{theorem}
\begin{proof}
  Recall that if two functions have the same code they are extensionally equal.
  By $\CT$, $\lambda f c. c \sim f$ is a total relation.
  Using $\WCN$ for this relation and $\lambda x.\;0$ yields a list $L$ and a code $c$ s.t.\ $\forall g.\;\map~g~L = [0, \dots, 0] \to c \sim g$.

  The functions $\lambda x.\;0$ and $\lambda x.\;\ite{x \in L}{0}{1}$ both fulfil the hypothesis and thus have the same code -- a contradiction since they are not extensionally equal. 
\end{proof}

\section{Conclusion}
\label{sec:conclusion}

In this paper we surveyed the known connections of axioms in Coq's type theory, a constructive type theory with a separate, impredicative universe of propositions, with a special focus on Church's thesis $\CT$ and formulations of axioms in terms of notions of synthetic computability.
Furthermore, all results are mechanised in the Coq proof assistant.

In constructive mathematics, countable choice is often silently assumed, as critised e.g.\ by Richman~\cite{richman2000fundamental,richmanConstructiveMathematicsChoice2001}.
In contrast, constructive type theory with a universe of propositions seems to be a suitable base system for matters of constructive (reverse) mathematics sensitive to applications of countable choice.
Due to the separate universe of propositions, such a constructive type theory neither proves countable nor dependent choice, allowing equivalences like the one in \Cref{thm:WKL_equivs} to be stated sensitively to choice.
We conjecture that \Cref{coq:semi_decidable_AC} deducing $\semidecidable\textsf{-}\AC_{X,\nat}$ and $\enumerable\textsf{-}\AC_{\nat,X}$ directly from $\decidable\textsf{-}\AC_{X,\nat}$ cannot be significantly strengthened.
The proof of $\decidable\textsf{-}\AC_{X,\nat}$ in turn crucially relies on a large elimination principle for $\exists n.\;f n = \btrue$ (\Cref{coq:mu_nat}). %
The theory of \cite{berger2012weak} proves $\decidable\textsf{-}\AC_{\nat,\bool}$ and thus likely also $\semidecidable\textsf{-}\AC_{\nat,\bool}$.

Predicative Martin-Löf type theory proves $\AC$ and type theories with propositional truncation and a semantic notion of (homotopy) propositions prove $\AUC_{\nat,\bool}$, thus $\LEM$ suffices to disprove $\CT$ for both these flavours of type theory.
Based on the current state of knowledge in the literature it seems likely that $\semidecidable\textsf{-}\AC_{\nat,\bool}$ and $\LEM$ together do not suffice to disprove $\CT$, which seems to require at least classical logic of the strength of $\LLPO$ and a choice axiom for \textit{co}-semi-decidable predicates.
Thus we conjecture that a consistency proof of e.g.\ $\LEM \land \CT$ might be possible for Coq's type theory.

Another advantage of basing constructive investigations on constructive type theory is that implementations of type theory in proof assistants already exist.
For this paper, mechanising the results in Coq was tremendously helpful in keeping track of all details.
For example, many of the presented proofs are very sensitive to small changes in formulations, and Coq actually helped in understanding the proofs and getting them right.

Besides consistency, another interesting property of axioms is admissibility.
For instance, Pédrot and Tabareau~\cite{PedrotMP} prove $\MP$ admissible in constructive type theory.
$\CT$ seems to be admissible in constructive type theory in the sense that for every defined function $f : \nat \to \nat$ one can define a program in a model of computation with the same input output behaviour, as witnessed by the certifying extraction for a fragment of Coq to the $\lambda$-calculus~\cite{forster_et_al:LIPIcs:2019:11072}.
An admissibility proof of $\CT$ could then serve as a theoretical underpinning of the Coq library of undecidability proofs~\cite{forster2020coq}.
However, any formal admissibility proof would have to deal with the intricacies of Coq's type theory.
It would be interesting to investigate whether Letouzey's semantic proof for the correctness of type and proof erasure~\cite{LetouzeyPhd} can be connected with the mechanisation of meta-theoretical properties of Coq's type theory~\cite{sozeau:hal-02167423} in the MetaCoq project~\cite{sozeau2019coq}, yielding a mechanised admissibility proof for $\CT$ in Coq's type theory.

\bibliography{biblio.bib}

\appendix

\section{Modesty and Oracles}
\label{sec:modesty}

Using $\decidable\textsf{-}\AC_{\nat,\nat}$ from \Cref{coq:decidable_AC} allows proving a choice axiom w.r.t.\ models of computation, observed by Larchey-Wendling~\cite{larchey2017typing} and called ``modesty'' by Forster and Smolka~\cite{forster2017weak}.

\begin{lemma}
  Let $T$ be an abstract computation function. We have
  \[\forall c. (\forall n.\exists m k.\;T c n k = \Some m) \to \exists f : \nat \to \nat. \forall n.\exists k.~T c n k = \Some (f n) \]
\end{lemma}

That is, if $c$ is the code of a function inside the model of computation which is provably total, the total function can be computed outside of the model.
This modesty principle simplifies the mechanisation of computability theory in type theory as e.g.\ in~\cite{forster2019call}.
For instance, it allows to prove that defining decidability as ``a total function in the model of computation deciding the predicate'' and as ``a meta-level function deciding the predicate which is computable in the model of computation'' is equivalent.

However, the modesty principle prevents synthetic treatments of computability theory based on oracles.
Traditionally, computability theory based on oracles is formulated using a computability function $T_p$, s.t.\ for $p : \nat \to \Prop$ there exists a code $c_p$ representing a total function s.t.\ $\forall n. (\exists k. T c_p n k = \Some 0) \leftrightarrow p n$.

Synthetically, we would now like to assume an abstract computability function for every $p$ as ``Church's thesis with oracles''.
``Church's thesis with oracles'' implies \CT, and we know that under \CT the predicate $\K_0$ is not decidable.
However, under the presence of $\decidable\textsf{-}\AC_{\nat,\nat}$ we can use $T_{\K_0}$ and obtain $c_{\K_0}$ which can be turned into a decider $f : \nat \to \bool$ for $\K_0$ using the choice principle above -- a contradiction.

\section{Coq mechanisation}

The Coq mechanisation of the paper comprises 4250 lines of code, with 3300 lines of proofs and 950 lines of statements and definitions, i.e.\ 77\% proofs.
The mechanisation is based on the Coq-std$\texttt{++}$ library~\cite{stdpp}, plus around 1500 additional lines of code with custom extensions to Coq's standard library which are shared with the Coq library of undecidability proofs~\cite{forster2020coq}.

The 4250 lines of the main development are distributed as follows:
The basics of synthetic computability (decidablility, semi-decidability, enumerability, many-one reductions) need 1150 lines of code.
The mechanisation of \Cref{sec:CT}, covering $\CT$, $\EA$, and $\EPF$, comprises 400 lines of code.
120 lines of codes are needed for the undecidability results of \Cref{sec:synth}.
\Cref{sec:kleene} and \Cref{sec:trees}, covering trees and in particular Kleene trees, need 1000 lines of code.
\Cref{sec:cont} on continuity is mechanised in 800 lines.
The rest, i.e.\ \Cref{sec:ext,sec:class,sec:russ,sec:choice}, needs 750 lines of code.

No advanced mechanisation techniques were needed.
Discreteness and enumerability proofs for types were eased using type classes to assemble proofs for compound types such as $\List \bool \times \option \nat$, as already done in~\cite{forster2019synthetic}.
Defining the notions of $\equiv_{A \to B}$, $\equiv_{A \to \Prop}$, and so on was made possible by using type classes as well.

The technically most challenging mechanised proofs correspond to Lemmas \ref{coq:F_inj} - \ref{coq:continuous_G}, i.e.\ prove $\KT \to \Homeo(\bool^\nat, \nat^ \nat) \land \Homeo(\nat^\nat, \bool^ \nat)$.
For these proofs, lots of manipulation of prefixes of lists was needed, and while the functions $\texttt{firstn}$ and $\texttt{dropn}$ are defined in Coq's standard library, the very useful lemmas of Coq-std$\texttt{++}$ where needed to make the proofs feasible.

In the development of this paper, the Coq proof assistant, while also acting as proof checker, was truly used as an assistant:
Lots of proofs were developed and understood directly while working in Coq rather than on paper, allowing to identify for instance the equivalent characterisations of $\LLPO$, $\MP$, and $\WKL$ as in \Cref{lem:LLPO_equivs} (\lipicsNumber{5}), \Cref{coq:MP_to_MP_semidecidable} (\lipicsNumber{5}), and \Cref{coq:WKL_to_LLPO} (\lipicsNumber{3,4}), which are hard to observe on paper because lots of bookkeeping for side-conditions would have to be done manually then.

\end{document}